\documentclass[1p]{elsarticle}

\usepackage{amssymb}
 \usepackage{amsthm}

\journal{CNSNS}

\usepackage{tkz-euclide}
\usepackage{relsize}
\usetikzlibrary{braids}
\usetikzlibrary{arrows.meta,shapes.misc,decorations.pathmorphing,calc}
\usetikzlibrary{decorations.markings}
\usetikzlibrary{arrows.meta,quotes}
\tikzset{fontscale/.style = {font=\relsize{#1}}
    }
    \newtheorem{definition}{Definition}
\newtheorem{theorem}{Theorem}
\newtheorem{lemma}{Lemma}
\newtheorem{corollary}{Corollary}
\newcommand{\be}{\begin{equation}}
\newcommand{\en}{\end{equation}}

\def\d{{\rm d}}
\def\uv{{\mathbf{u}}}

\def\rv{{\mathbf{r}}}

\def\xv{{\mathbf{x}}}
\def\yv{{\mathbf{y}}}
\def\nv{{\mathbf{n}}}

\def\ev{{\mathbf{e}}}
\def\Bv{{\mathbf{B}}}
\def\Ev{{\mathbf{E}}}
\def\Cv{{\mathbf{C}}}

\def\Av{{\mathbf{A}}}

\def\cl{{\mathcal{L}}}

\def\grad{\boldsymbol\grad}

\def\grad{{\rm grad}\, }

\begin{document}

\begin{frontmatter}



\title{On self and mutual winding helicity}


\author[gla]{Simon Candelaresi}
\author[dun]{Gunnar Hornig}
 \author[gla]{David MacTaggart\footnote{Corresponding author: david.mactaggart@glasgow.ac.uk}} 
 \author[gla]{Radostin D. Simitev}

\address[gla]{School of Mathematics and Statistics, University of Glasgow, Glasgow, G12 8SQ, UK}

\address[dun]{Division of Mathematics, University of Dundee, Dundee, DD1 4HN, UK}

\begin{abstract}
The topological underpinning of magnetic fields connected to a planar boundary is naturally described by field line winding. This observation leads to the definition of \emph{winding helicity}, which is closely related to the more commonly calculated \emph{relative helicity}. Winding helicity, however, has several advantages, and we explore some of these in this work. In particular, we show, {by splitting the domain into distinct subregions,} that winding helicity can be decomposed naturally into ``self'' and ``mutual'' components and that these quantities can be calculated, in practice, for magnetic fields with complex geometries and topologies. Further, winding provides a unified topological description from which known expressions for self and mutual helicity can be readily derived and generalized. We illustrate the application of calculating self and mutual winding helicities in a simulation of an evolving magnetic field with non-trivial field line topology.

\end{abstract}


\end{frontmatter}

\section{Introduction}

\subsection{Classical helicity}
Magnetic helicity is an invariant quantity of ideal magnetohydrodynamics (MHD) that encodes information about the topology of magnetic field lines. The initial understanding of magnetic helicity was developed for closed magnetic fields \cite{1,2}. In a simply connected domain $\Omega$ (see \cite{3} for extensions to multiply connected domains), the magnetic helicity has the form
\be\label{hel_class}
H=\int_{\Omega}\Av\cdot\Bv\,\d V,
\en
where $\Bv$ is the magnetic induction (hereafter referred to as the magnetic field) and $\Av$ is the magnetic vector potential, satisfying $\Bv=\nabla\times\Av$. If $\Bv\cdot\nv=0$ on $\partial\Omega$, where $\nv$ is the unit normal to $\partial\Omega$, $H$ from equation (\ref{hel_class}) is gauge invariant. 

The topological interpretation of $H$ can be seen by writing the vector potential in terms of the {Biot-Savart} operator \cite{2,4,5},
\be\label{biot-savart}
\Av(\xv) = \frac{1}{4\pi}\int_{\Omega}\Bv(\yv)\times\frac{\xv-\yv}{|\xv-\yv|^3}\,\d^3y.
\en
Substituting equation (\ref{biot-savart}) into equation (\ref{hel_class}) leads to
\be\label{hel_bs}
H=\frac{1}{4\pi}\int_{\Omega\times\Omega}\Bv(\xv)\cdot\Bv(\yv)\times\frac{\xv-\yv}{|\xv-\yv|^3}\,\d^3x\,\d^3y,
\en
which is a representation of the Gauss linking number weighted by magnetic flux.  

 Let $\Omega=\sqcup_i\Omega_i\subset\mathbb{R}^3$, i.e. $\Omega$ is composed of disjoint subdomains $\Omega_i$. In each subdomain there is a magnetic field $\Bv_i$ with $\Bv_i\cdot\nv=0$ on $\partial\Omega_i$. For demonstration purposes, we will, for the moment, consider just two disjoint subdomains. Let $\Bv_1$ be in $\Omega_1$ and $\Bv_2$ be in $\Omega_2$. With $\Bv=\Bv_1+\Bv_2$ in $\Omega=\Omega_1\sqcup\Omega_2$, equation (\ref{hel_bs}) can be expanded as {(see \cite{24})}
\be\label{hel_exp}
H_{\Omega_1\sqcup\Omega_2}(\Bv) = H_{\Omega_1}(\Bv_1)+H_{\Omega_2}(\Bv_2)+2H(\Bv_1,\Bv_2),
\en
where the first two terms on the right-hand side of equation (\ref{hel_exp}) are given by equation (\ref{hel_bs}) in the domains $\Omega_1$ and $\Omega_2$ respectively, and 
\be\label{hel_class_mutual}
H(\Bv_1,\Bv_2) = \frac{1}{4\pi}\int_{\Omega_1\times\Omega_2}\left(\Bv_1(\xv)\times\Bv_2(\yv)\cdot\frac{\xv-\yv}{|\xv-\yv|^3}\right)\,\d^3x\,\d^3y.
\en
The quantities $H_{\Omega_1}$ and $H_{\Omega_2}$ are called the \emph{self helicities} and depend on the internal magnetic field topology as well as the domain topology \cite{3}. Equation (\ref{hel_class_mutual}) represents the \emph{mutual helicity} which takes into account the linkage of different subdomains. Indeed, in Moffatt's seminal work on helicity, it was the mutual helicity of magnetic (and vortex) fields that he described. 

\subsection{Relative helicity}
For magnetic fields which are not closed but have a {non-zero} normal component on the domain boundary, the classical helicity in equation (\ref{hel_class}) is no longer gauge invariant. In order to have a gauge invariant form of helicity in this situation, we have to consider a relative measure of helicity that compares two different magnetic fields with the same normal components on the boundary. Let $\Bv_1$ and $\Bv_2$ be such magnetic fields, then the gauge invariant relative magnetic helicity \cite{6,7} is given by
\be\label{rel_hel}
H_R = \int_{\Omega}(\Av_1+\Av_2)\cdot(\Bv_1-\Bv_2)\,\d^3x,
\en
where $\Av_1$ and $\Av_2$ are the vector potentials of $\Bv_1$ and $\Bv_2$ respectively. This measure of helicity has been used extensively in solar active region applications \cite{9,19}. Although gauge invariant decompositions of relative helicity {have been proposed} \cite{8}, there is no general self-mutual decomposition for relative helicity as for classical helicity in (\ref{hel_exp}). This is an important point, particularly with respect to what is reported in the literature, which we will return to shortly. First, we describe an alternative measure of helicity for open magnetic fields which does admit a self-mutual decomposition.  

\section{Winding helicity}
Henceforth, $\Omega$ will represent a simply connected domain between two horizontal planes at heights $z=0$ and $z=h$. On these planes the magnetic field can have $\Bv\cdot\nv\ne 0$. The side boundaries (any boundaries other than the two planes) will be magnetic surfaces. In what follows, the results still hold if the horizontal planes are also magnetic surfaces. Each disjoint subdomain of $\Omega$ is labelled $\Omega_i$ and has a magnetic surface as its boundary apart form where it is in contact with the horizontal planes. Subdomains can be connected to one, two or neither of the planes. The subdomains can be simply connected (no holes) or multiply connected (contain holes).

Rather than seeking a general gauge invariant measure of helicity, we choose one with a particular gauge that provides a clear topological interpretation (just as the Coulomb gauge, with the Biot-Savart operator, does for classical helicity). We call this measure of helicity \emph{winding helicity} \cite{10,11,12,13} since it represents the average pairwise winding between all local portions of field lines, weighted by magnetic flux.

\begin{definition}
The winding gauge $\Av^W$ satisfies $\nabla^{\perp}\cdot\Av^W=0$, where $\nabla^{\perp}=(\partial/\partial x, \partial/\partial y, 0)^T$ and can be written as
\be\label{wind_gauge}
\Av^W(\Bv)(x_1,x_2,z) = \frac{1}{2\pi}\int_{S_z} \Bv(y_1,y_2,z)\times\frac{\xv-\yv}{|\xv-\yv|^2}\,\d^2y,
\en
where $S_z$ is a horizontal slice of $\Omega$ at height $z$, $\xv = (x_1(z),x_2(z),0)^T$ and $\yv=(y_1(z),y_2(z),0)^T$.
\end{definition}
\noindent The position vectors $\xv$ and $\yv$ lie on horizontal planes at different heights throughout $\Omega$, so the above notation should be interpreted as: given some height $z$, there are horizontal vectors $\xv(z)$ and $\yv(z)$ pointing to magnetic field lines intersecting the horizontal plane at height $z$. {This situation is illustrated in Figure \ref{planes}. Note that field lines need not be monotonically increasing in height and can intersect the same plane in multiple locations. A detailed comparison of the winding gauge with alternative gauges has been performed by Prior and Yeates \cite{10}, who highlight the advantage of the winding gauge.}

\begin{figure}
\centering
\begin{tikzpicture}
\draw[thick] (-5,5)--(5,5);
\draw[thick] (-5,0)--(5,0);
\draw[thick] (-5,5)--(-5,0);
\draw[thick] (5,5)--(5,0);

\node at (4, 0.5) [right] {$S_z$};

\draw [fill,red] (-1,2.5) circle [radius=.05];
\draw [fill,blue] (3.5,4) circle [radius=.05];

\draw[thick,dashed] (-1,2.5)--(3.5,4);

\draw[thick,->] (-1.5,1)--(-1,2.5);
\draw[thick,->] (-1.5,1)--(3.5,4);

\node at (-1.5,1) [below] {$O$};

\draw[thick,dashed] (-1,2.5)--(0.6,2.5);

\node at (-0.05, 2.7) [right] {$\theta$};

\node at (1, 2.2) [right] {$\mathbf{x}$};

\node at (-1.8, 1.8) [right] {$\mathbf{y}$};

\end{tikzpicture}
\caption{\label{planes}A projection looking down vertically on a horizontal plane $S_z$ with the intersection of two  field lines marked by blue and red points. The horizontal position vectors $\xv$ and $\yv$ of these points are indicated, together with the polar angle $\theta$ of $\xv-\yv$.}
\end{figure}
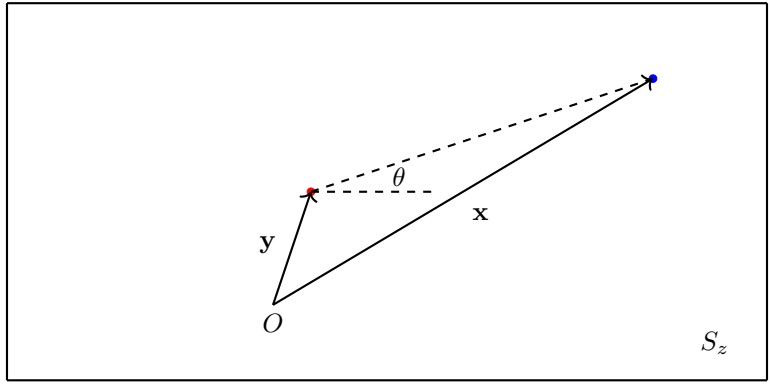

The winding gauge is then used to construct the winding helicity as follows.
\begin{definition}
The winding helicity is defined as
\be\label{wind_hel}
H^W = \int_{\Omega}\Av^W\cdot\Bv\,\d^3x
=\frac{1}{2\pi}\int_0^h\int_{S_z\times S_z}\Bv(\xv)\cdot\Bv(\yv)\times\frac{\xv-\yv}{|\xv-\yv|^2}\,\d^2x\,\d^2y\,\d z.
\en
\end{definition}
\noindent In order to {simplify the} notation, the fact that the vectors $\xv$ and $\yv$ vary in $z$ will be suppressed and assumed implicitly. This should not cause confusion as the integrals are performed in horizontal planes first (each with a different value of $z$) and are then integrated in height.

Notice the similarities between (\ref{biot-savart}) and (\ref{wind_gauge}) and (\ref{hel_bs}) and (\ref{wind_hel}). {Prior and Yeates \cite{10}} prove that the winding helicity can be expressed directly in terms of the winding of field lines, i.e.
\be\label{wind_hel2}
H^W =\frac{1}{2\pi}\int_0^h\int_{S_z\times S_z}\frac{\d\theta(\xv,\yv)}{\d z}B_z(\xv)B_z(\yv)\,\d^2x\,\d^2y\,\d z,
\en
where $\theta(\xv,\yv)$ 
is the polar angle of  $\xv-\yv$ and is given by
\be\label{theta}
\theta(\xv,\yv) = {\rm arctan}\left[\frac{(\xv-\yv)\cdot\ev_y}{(\xv-\yv)\cdot\ev_x}\right],
\en
where $\{\ev_x,\ev_y,\ev_z\}$ is the standard Cartesian basis. It is from equation (\ref{wind_hel2}) that the winding interpretation of helicity is made clear.

Winding helicity has a close relationship to relative helicity and {Prior and Yeates \cite{10}} also derive the following result.
\begin{lemma}\label{lemma}
Let $\Bv$ and $\Bv'$ be two magnetic fields in $\Omega$ with the same normal components on the boundary. Then the relative and winding helicities are related by
\be
H_R(\Bv,\Bv') = H^W(\Bv) - H^W(\Bv').
\en
\end{lemma}

\section{\label{sec:self_mut}Helicity decomposition}
We are now in a position to state our main result - the decomposition of winding helicity into self and mutual components.

\begin{theorem}\label{sm_thrm}
The winding helicity in the domain $\Omega=\sqcup_i\Omega_i$ satisfies the decomposition
\be\label{sm_decom}
    H^W=\sum_i H_i^W + \sum_{{i,j \atop (i\ne j)}}H^W_{ij},
\en
where the self helicities are
\be\label{self}
H^W_i=\frac{1}{2\pi}\int_0^h\int_{{S_z\cap\Omega_i\times \atop S_z\cap\Omega_i}}\frac{\d\theta(\xv,\yv)}{\d z}B_z(\xv)B_z(\yv)\,\d^2x\,\d^2y\,\d z,
\en
and the mutual helicities are
\be\label{mutual}
H^W_{ij}=\frac{1}{2\pi}\int_0^h\int_{{S_z\cap\Omega_i\times \atop S_z\cap\Omega_j}}\frac{\d\theta(\xv,\yv)}{\d z}B_z(\xv)B_z(\yv)\,\d^2x\,\d^2y\,\d z,
\en
and satisfy $H^W_{ij}=H^W_{ji}$.  In equation (\ref{self}), the position vectors $\xv$ and $\yv$ refer to different points in the same subdomain $S_z\cap\Omega_i$. In equation (\ref{mutual}), the position vectors $\xv$ refer to points in $S_z\cap\Omega_i$ and the position vectors $\yv$ refer to those in $S_z\cap\Omega_j$.
\end{theorem}

\begin{proof}
The proof follows from the winding gauge having similar properties to the Coulomb gauge (the {Biot-Savart} operator). We first show that $\Av^W$ is symmetric by considering two different fields $\Bv_1$ and $\Bv_2$. Now
\begin{eqnarray}\label{symmetric}
&&\int_{\Omega}\Av^W(\Bv_1)\cdot\Bv_2\,\d V \nonumber\\ 
&&=\frac{1}{2\pi}\int_0^h\int_{S_z\times S_z}\Bv_2(\xv)\cdot\frac{\Bv_1(\yv)\times(\xv-\yv)}{|\xv-\yv|^2}\,\d^2y\,\d^2x\,\d z \nonumber \\ 
&&= -\frac{1}{2\pi}\int_0^h\int_{S_z\times S_z}\Bv_1(\yv)\cdot\frac{\Bv_2(\xv)\times(\xv-\yv)}{|\xv-\yv|^2}\,\d^2y\,\d^2x\,\d z \nonumber \\ 
&&= \frac{1}{2\pi}\int_0^h\int_{\Omega}\Bv_1(\yv)\cdot\frac{\Bv_2(\xv)\times(\yv-\xv)}{|\yv-\xv|^2}\,\d^2x\,\d^2y\,\d z \nonumber \\ 
&&=\int_{\Omega}\Av^W(\Bv_2)\cdot\Bv_1\,\d V
\end{eqnarray}
From the linearity of integration and the above operations on vectors (scalar and vector products), it follows that $\Av^W$ is linear in the sense that
\be\label{linear}
\int_{\Omega}\Av^W(\Bv_1+\Bv_2)\cdot\Bv\,\d V = \int_{\Omega}\Av^W(\Bv_1)\cdot\Bv\,\d V + \int_{\Omega}\Av^W(\Bv_2)\cdot\Bv\,\d V.
\en
With properties (\ref{symmetric}) and (\ref{linear}), we can now decompose the winding helicity as
\be\label{decompose}
 \int_{\Omega}\Av^W\left(\sum_i\Bv_i\right)\cdot\sum_i\Bv_i\,\d V=
 \sum_i\int_{\Omega_i}\Av^W(\Bv_i)\cdot\Bv_i\,\d V 
 +\sum_{{i,j \atop (i\ne j)}}\int_{\Omega_i\cup\Omega_j}\Av^W(\Bv_i)\cdot\Bv_j\,\d V.
\en
The integrals in the first and second terms on the right-hand side of equation (\ref{decompose}) correspond to equations (\ref{self}) and (\ref{mutual}) respectively. The derivation of equations (\ref{self}) and (\ref{mutual}) now proceeds almost exactly as in \cite{10} and will be omitted for brevity. The only difference that we need to consider here is the domain of integration. For the self helicity, only the magnetic field of one subdomain $\Omega_i$ is needed and, when determining $\Av^W$, the intersection of this domain with the plane $S_z$ is required. For the mutual helicity, we require the fields in two subdomains $\Omega_i$ and $\Omega_j$ and their intersection with $S_z$ for calculating the winding gauge. Finally, from property (\ref{symmetric}), it follows immediately that $H^W_{ij}=H^W_{ji}$.
\end{proof}

\section{Relationship to previous work on relative helicity}
Previous works \cite{14,15,16,17,18} have developed self-mutual decomposition formulae, analogous to equation (\ref{sm_decom}), for thin discrete flux tubes. These formulae are highly appealing as they reduce complicated helicity calculations to simple problems of Euclidean geometry. Theorem \ref{sm_thrm} generalizes these previous results, allowing self and mutual helicities to be found in domains that are more complicated than discrete tubes. For example, equations (\ref{self}) and (\ref{mutual}) can be applied to subdomains that are multiply connected, i.e. contain holes. It will be shown below that previous results can be derived from our more general setup, which depends on magnetic winding \cite{13} as the underlying topological structure. This opens a more direct and coherent way of deriving previously known results, in addition to providing a practical approach to calculating self and mutual helicities in subdomains with complex geometries and topologies.

At the end of Section 1, we stated the relative helicity, in the form of equation (\ref{rel_hel}), does not possess a self-mutual decomposition like classical helicity. This statement appears to contradict previous works \cite{14,15,16,17,18}, for which self-mutual decomposition formulae have been found for relative helicity. However, our work and theirs can be unified by taking into account some basic considerations. {Berger \cite{14}} first derived the self and mutual helicity formulae for magnetic flux tubes between two planes. In this work, he assumes that the relative helicity can be written in the form of equation (\ref{wind_hel2}). Although this is formally correct, it comes with a caveat. From Lemma \ref{lemma}, if winding and relative helicities are to be equated then the winding helicity of the reference field $\Bv'$ must be zero. Thus, in the literature, when self and mutual helicities are reported with respect to relative helicity, this fact must be assumed implicitly.

\subsection{Thin tube approximations}

\subsubsection{Monotonically increasing in height}
As mentioned above, equations (\ref{self}) and (\ref{mutual}) reduce naturally to the self and mutual helicity formulae used to approximate the (relative) helicity of thin and discrete flux tubes \cite{18}. To illustrate this, we will first consider the case where magnetic flux tubes connecting the two planes are monotonically increasing in height \cite{14,15,18}. An illustration of this for two flux tubes is given in Figure \ref{monotonic}.
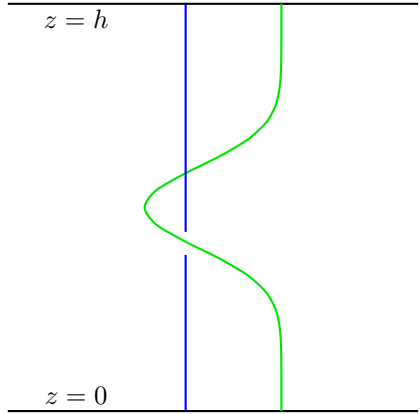
\begin{figure}
\centering
\begin{tikzpicture}[scale=0.9]
\draw[thick] (-3,3)--(3,3);
\draw[thick] (-3,-3)--(3,-3);
\draw[thick,smooth,black!10!green] plot[variable=\x,domain=-3:3] ({-2*exp(-1.4*\x*\x)+1},{\x});

\draw[thick,blue] (-0.4,-3)--(-0.4,-0.7);
\draw[thick,blue] (-0.4,-0.36)--(-0.4,3);

\node at (-2,-2.5) [anchor=north]{$z=0$};
\node at (-2,2.5) [anchor=south]{$z=h$};
\end{tikzpicture}
\caption{\label{monotonic}An illustration of two discrete flux tubes monotonically increasing in the $z$-direction between two horizontal planes. Each tube has internal field line structure leading to self helicity. The winding of the flux tubes around each other leads to mutual helicity.}
\end{figure}

Suppose that we have a thin flux tube $\Omega_i$ for which $\d\theta/\d z$ only varies in height, i.e. at a particular height $z$, the pairwise winding for all pairs of field lines is the same. The self helicity (\ref{self}) then reduces to
\begin{eqnarray}\label{tube_self}
H^W_i&=&\frac{1}{2\pi}\int_0^h\int_{{S_z\cap\Omega_i\times \atop S_z\cap\Omega_i}}\frac{\d\theta(\xv,\yv)}{\d z}B_z(\xv)B_z(\yv)\,\d^2x\,\d^2y\,\d z, \nonumber \\
&=& \frac{1}{2\pi}\int_0^h\frac{\d\theta}{\d z}\,\d z \int_{{S_z\cap\Omega_i\times \atop S_z\cap\Omega_i}}B_z(\xv) B_z(\yv)\,\d^2x\,\d^2y, \nonumber \\
&=& \mathcal{L}_{\rm self}\Phi_i^2,
\end{eqnarray} 
where 
\be\label{l_self}
\mathcal{L}_{\rm self} = \frac{1}{2\pi}\int_0^h\frac{\d\theta}{\d z}\,\d z
\en
is the winding of any pair of field lines in $\Omega_i$ and
\be
\Phi_i = \int_{S_z\cap\Omega_i}B_z(\xv)\,\d^2x,
\en
is the magnetic flux in the tube. The winding $\mathcal{L}_{\rm self}$ depends on the internal twisting of field lines and the writhe of the tube.

If another thin flux tube $\Omega_j$ is present in $\Omega$ and obeys the same assumptions as $\Omega_i$, then the mutual helicity becomes
\begin{eqnarray}\label{tube_mutual}
H^W_{ij}&=&\frac{1}{2\pi}\int_0^h\int_{{S_z\cap\Omega_i\times \atop S_z\cap\Omega_j}}\frac{\d\theta(\xv,\yv)}{\d z}B_z(\xv)B_z(\yv)\,\d^2x\,\d^2y\,\d z, \nonumber \\
&=& \frac{1}{2\pi}\int_0^h\frac{\d\theta}{\d z}\,\d z \int_{S_z\cap\Omega_i}B_z(\xv)\,\d^2x \int_{S_z\cap\Omega_j}B_z(\yv)\,\d^2y, \nonumber \\
&=& \mathcal{L}_{\rm mutual}\Phi_i\Phi_j.
\end{eqnarray}
Assuming that $\Omega_i$ and $\Omega_j$ are suitably thin in order to treat them as strings rather than tubes, the mutual pairwise winding can be written as 
\be
\mathcal{L}_{\rm mutual} = \frac{1}{2\pi}[\theta(z=h)-\theta(z=0)] + N,
\en
where $N$ is the full number of rotations of $\Omega_i$ about $\Omega_j$. Equations (\ref{tube_self}) and (\ref{tube_mutual}) can be found, for example, in \cite{18}.

\subsubsection{Both ends connected to the same plane}
For the case of non-monotonically increasing flux tubes, such as those connected to the same plane, a similar approach can be adopted to that described above. Now, however, the concept of winding has to be generalized in order to account for magnetic field lines that are not monotonically increasing in $z$. Each field line is split into monotonically increasing (decreasing) segments and these are assigned positive (negative) weights to account for their directions. Thus a generalized version of the pairwise winding, first presented by {Berger and Prior \cite{20}}, can be defined as follows. 
\begin{definition}
Let two field lines $\xv$ and $\yv$ have $n$ and $m$ distinct turning points respectively, that is points where $\d x_z/\d z = 0$ with $\xv\cdot\ev_z=x_z$ or  $\d y_z/\d z = 0$ with $\yv\cdot\ev_z = y_z$. Let $\xv$ be partitioned into $n + 1$ regions and $\yv$ into $m + 1$ regions so that curve sections $\xv_i$ and $\yv_j$ share a mutual z-range $[z_{ij}^{\min}, z_{ij}^{\max}]$ in each region. The total winding is defined as the sum of the weighted pairwise winding in each region,
\be\label{wind2}
\cl(\xv,\yv) = \sum_{i=1}^{n+1}\sum_{j=1}^{m+1}\frac{\sigma(\xv_i)\sigma(\xv_j)}{2\pi}\int_{z_{ij}^{\rm min}}^{z_{ij}^{\rm max}} \frac{\d}{\d z}\theta(\xv_i(z),\yv_j(z))\,\d z,
\en
where $\sigma(\xv_i)$ is an indicator function marking where the curve section $\xv_i$ moves up or down in $z$, i.e.
\be\label{sigma}
\sigma(\xv_i) = \left\{\begin{array}{ccc}
1 & {\rm if} \quad & \d x_z/\d z > 0, \\
-1 & {\rm if} \quad &\d x_z /\d z < 0, \\
0 & {\rm if} \quad &\d x_z /\d z = 0. \end{array}\right.
\en
\end{definition}
\noindent It is the generalized form of winding given in equation (\ref{wind2}) that acts as the underlying topological structure of winding helicity  \cite{10,13}. This quantity, of course, reduces to the standard winding formula (e.g. equation \ref{l_self}) when the field lines are monotonically increasing in $z$.

We now consider a specific example to indicate how the thin tube approximation of mutual helicity can be derived for two flux tubes anchored at the same planar boundary. We present the derivation in two ways in order to show more clearly the winding interpretation connected to previously derived helicity formulae. In our first approach, we will calculate the winding of one flux tube relative to the footpoints of the other, which requires a simpified geometric construction. In the second approach, which is more general, we will calculate the full pairwise winding of both flux tubes. The purpose of employing these two approaches is to show that the mutual helicity can have two alternative geometric interpretations.

We begin by outlining the basic geometry required for both calculations. Figure \ref{two_arches} displays two semi-circular flux tubes connecting to the same horizontal plane.
\begin{figure}[h]
\centering
\begin{tikzpicture}[scale=0.9]
\draw[thick] (-5,0) --(5,0);

\draw[thick,blue,->] ([shift=(180:3cm)]-1,0) arc (180:160:2.5cm);
\draw[thick,blue] ([shift=(0:3cm)]-2,0) arc (0:180:2.5cm);
\draw[thick,black!20!green] ([shift=(0:3cm)]1.5,0) arc (0:128:3cm);
\draw[thick,black!20!green] ([shift=(180:3cm)]1.5,0) arc (180:133:3cm);
\draw[thick,black!20!green,->] ([shift=(180:3cm)]1.5,0) arc (180:160:3cm);

\fill (-1.5,2.5)  circle[radius=2pt];

\fill (1.5,3)  circle[radius=2pt];

\fill[red] (-0.15,2.5)  circle[radius=2pt];

\fill[red] (3.15,2.5)  circle[radius=2pt];
\draw[thick,dashed,red] (-0.15,2.5)--(3.15,2.5);

\draw[thick, dashed,purple] (-4,0)--(-4,2.5);

\fill[purple] (-4,2.5)  circle[radius=2pt];
\node[purple] at (-4,3.1) [anchor=north]{$h_1$};

\node[purple] at (-4.45,1.6) [anchor=north]{$\alpha_1'$};

\node at (-1.5,3.1) [anchor=north]{$h_1$};
\node at (1.5,3.6) [anchor=north]{$h_2$};

\node[blue] at (-3.2,1.5) [anchor=north]{$\alpha_1$};
\node[blue] at (1,1.5) [anchor=north]{$\alpha_2$};

\node[black!20!green] at (-1.6,1.5) [anchor=north]{$\beta_1$};
\node[black!20!green] at (3.9,1.5) [anchor=north]{$\beta_2$};

\node[red] at (1.5,2) [anchor=south]{$h_1$};
\end{tikzpicture}
\caption{\label{two_arches}An illustration of two discrete flux tubes $\alpha = \alpha_1\cup\alpha_2$ and $\beta = \beta_1\cup\beta_2$ connected to the same plane. These tubes have maximum heights of $h_1$ and $h_2$ respectively. The vertical line $\alpha_1'$ is used for the simplified calculation of winding, relative to a footpoint, instead of $\alpha_1$. {A vertical projection of the flux tubes, as seen from above, is shown in Figures \ref{angles} and \ref{pw_angles}.}}
\end{figure}
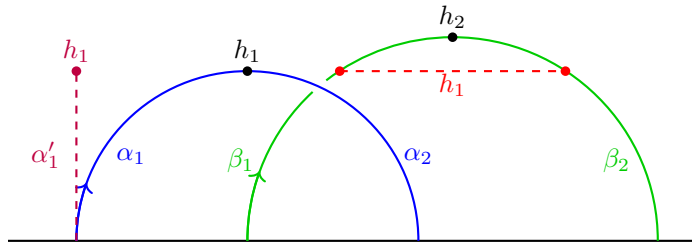
The turning points of the flux tubes are labelled with their heights $h_1$ and $h_2$. These points divide the tubes into monotonically increasing/decreasing regions. The blue tube (with flux $\Phi_i$) is split into parts, $\alpha_1$ and $\alpha_2$ on either side of the turning point. Similarly, the green tube (with flux $\Phi_j$) is split into $\beta_1$ and $\beta_2$. Since $h_2>h_1$, the part of the green tube above $h_2$  will not contribute to the winding and the helicity (it is effectively weighted by zero magnetic flux). The upper boundary of where the green tube contributes to the pairwise winding is indicated by the red dashed line at height $h_1$.

Figure \ref{two_arches} also displays a vertical line $\alpha_1'$, of height $h_1$, anchored at the footpoint of $\alpha_1$. In the simplified calculation of winding, $\alpha'_1$ will be used instead of $\alpha_1$ (and, similarly, $\alpha_2'$, which is not shown, instead of $\alpha_2$). Each vertical line inherits the same magnetic field direction as its corresponding section of the flux tube.

With the basic geometry in place, we can now proceed to calculating the winding. Using equation (\ref{wind2}), this is done by comparing the winding of different sections of curves. We first focus on the simplified winding case and consider the pairings:
 $(\alpha_1',\beta_1)$, $(\alpha_1',\beta_2)$, $(\alpha_2',\beta_1)$ and $(\alpha_2',\beta_2)$. Figure \ref{angles}(a) displays the winding angles for each of the above pairings.

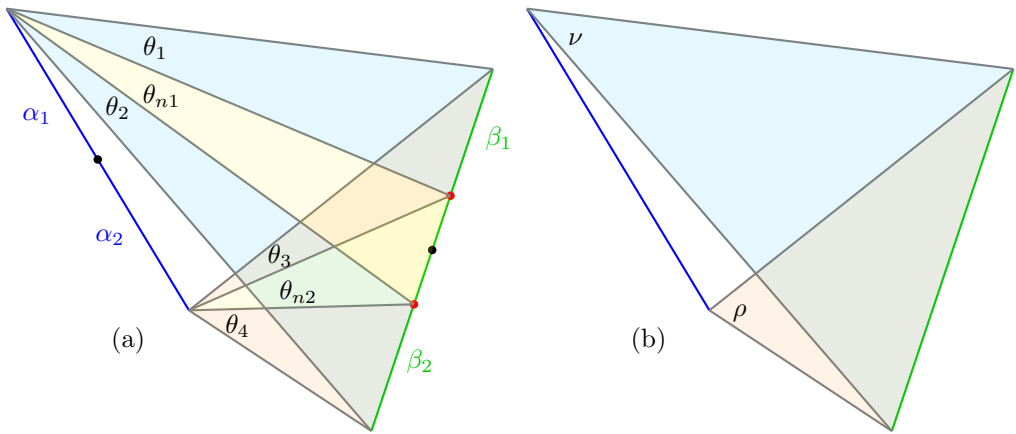
\begin{figure}[h]
\centering
{\begin{tikzpicture}[scale=0.8]
\draw[thick,blue] (-3,3)--(0,-2);
\draw[thick,black!20!green] (3,-4)--(5,2);

\fill (-1.5,0.5)  circle[radius=2pt];
\fill (4,-1)  circle[radius=2pt];

\fill[red] (3.7,-1.9)  circle[radius=2pt];
\fill[red] (4.3,-0.1)  circle[radius=2pt];

\draw[thick,gray] (-3,3)--(3,-4);
\draw[thick,gray] (-3,3)--(3.7,-1.9);

\draw[thick,gray] (-3,3)--(4.3,-0.1);
\draw[thick,gray] (-3,3)--(5,2);

\draw[thick,gray] (0,-2)--(5,2);
\draw[thick,gray] (0,-2)--(4.3,-0.1);

\draw[thick,gray] (0,-2)--(3.7,-1.9);
\draw[thick,gray] (0,-2)--(3,-4);

\fill[cyan,fill opacity=0.1] (-3,3)--(3,-4)--(3.7,-1.9); 
\fill[cyan,fill opacity=0.1] (-3,3)--(4.3,-0.1)--(5,2); 

\fill[orange,fill opacity=0.1] (0,-2)--(3,-4)--(3.7,-1.9); 
\fill[orange,fill opacity=0.1] (0,-2)--(5,2)--(4.3,-0.1); 

\fill[yellow,fill opacity=0.1] (0,-2)--(4.3,-0.1)--(3.7,-1.9); 
\fill[yellow,fill opacity=0.1] (-3,3)--(4.3,-0.1)--(3.7,-1.9);

\node[blue] at (-2.5,1.5) [anchor=north]{$\alpha_1$};
\node[blue] at (-1.3,-0.5) [anchor=north]{$\alpha_2$};

\node[black!20!green] at (5.1,1.2) [anchor=north]{$\beta_1$};
\node[black!20!green] at (3.8,-2.5) [anchor=north]{$\beta_2$};

\node at (0.8,-1.91) [anchor=north]{$\theta_4$};
\node at (1.5,-0.76) [anchor=north]{$\theta_3$};
\node at (1.8,-1.35) [anchor=north]{$\theta_{n2}$};

\node at (-0.55,2.7) [anchor=north]{$\theta_1$};
\node at (-1.17,1.7) [anchor=north]{$\theta_2$};
\node at (-0.45,1.9) [anchor=north]{$\theta_{n1}$};

\node at (-1,-2.5){(a)};

\end{tikzpicture}}\,{\begin{tikzpicture}[scale=0.8]
\draw[thick,blue] (-3,3)--(0,-2);
\draw[thick,black!20!green] (3,-4)--(5,2);

%

\draw[thick,gray] (-3,3)--(3,-4);
\draw[thick,gray] (0,-2)--(5,2);
\draw[thick,gray] (-3,3)--(5,2);
\draw[thick,gray] (0,-2)--(3,-4);

\fill[cyan,fill opacity=0.1] (-3,3)--(3,-4)--(5,2); 
\fill[orange,fill opacity=0.1] (0,-2)--(3,-4)--(5,2); 

\node at (-2.2,2.75) [anchor=north]{$\nu$};
\node at (0.5,-1.75) [anchor=north]{$\rho$};
\node at (-1,-2.5){(b)};
\end{tikzpicture}}
\caption{\label{angles}The winding angles of $\beta$ relative to the footpoints of $\alpha$. (a) shows the individual angles described in the text. (b) shows the sum of angles, as described in other works such as \cite{14,18}.}
\end{figure}
For the pairing $(\alpha_1',\beta_1)$, we have $\sigma(\alpha_1')=\sigma(\beta_1)=1$. The magnitude of the winding angle is $\theta_1$ but since it is found by moving clockwise, the actual winding angle is $-\theta_1$. Therefore, the winding of this part is $\cl(\alpha_1',\beta_1) = -\theta_1/(2\pi)$.  For the pairing $(\alpha_1',\beta_2)$, we now have $\sigma(\alpha_1')=1$ and $\sigma(\beta_2)=-1$. Also, the magnitude of winding angle $\theta_2$ is found by moving in the anti-clockwise direction and so the angle is equivalent to its magnitude. Hence, the winding for this part is $\cl(\alpha_1',\beta_2) = -\theta_2/(2\pi)$. 

There is a gap between the regions of  $\cl(\alpha_1',\beta_1)$ and $\cl(\alpha_1',\beta_2)$, indicated in Figure \ref{angles} by an angle of magnitude $\theta_{n1}$. Some care needs to be taken when integrating up to the point at height $h_2$ as the angle $\theta$ undergoes a discontinuous jump at this point. There is a contribution to the winding moving from $\alpha_1'$ to $\beta_1$ along the red dashed line, where $\theta$ jumps in value. This contribution is
\be
\cl({\rm jump}) = -\frac{1}{2\pi}\int_{\alpha_1'}^{\beta_1}\frac{\d}{\d z}\theta\d z = -\frac{1}{2\pi}[\theta]_{\alpha_1'}^{\beta_1} = -\frac{1}{2\pi}\theta_{n1},
\en
where the negative sign takes account of the $\sigma$-weighting and the direction of the angle, and the derivative is taken in the sense of distributions. Thus the total winding of $\beta=\beta_1\cup\{z=h_1\}\cup\beta_2$ about $\alpha_1'$ is
\be\label{nu}
\cl(\alpha_1',\beta) = -\frac{1}{2\pi}(\theta_1 + \theta_{n1} + \theta_2) = -\frac{1}{2\pi}\nu,
\en
where $\nu=\theta_1 + \theta_{n1} + \theta_2$, as shown in Figure \ref{angles}(b). By applying similar reasoning, we find 
\be\label{rho}
\cl(\alpha_2',\beta) = \frac{1}{2\pi}(\theta_3 + \theta_{n2} + \theta_4)=\frac{1}{2\pi}\rho,
\en
where $\rho=\theta_3 + \theta_{n2} + \theta_4$. The mutual helicity is found by adding these windings together and multiplying by the fluxes (assigning $\alpha'$ the flux $\Phi_i$ as it represents the footpoint of $\alpha$),
\be\label{wright_berger}
H^W_{ij} = \frac{1}{2\pi}(\rho-\nu)\Phi_i\Phi_j.
\en
This is the thin-tube formula for the mutual helicity of two non-crossing tubes that can be found in several works \cite{14,18}. It is straightforward to check that $H_{ji}=H_{ij}$ by perfoming the above procedure but starting from $\beta$.  The calculation is simpler than the one above since the boundary point between $\alpha_1$ and $\alpha_2$ does not require the consideration of a `winding jump'. Simple Euclidean geometry then shows that the two mutual helicities are equivalent (as described in \cite{14}). This approach shows that the geometrical meaning of the mutual helicity in equation (\ref{wright_berger}) can be interpreted as the winding of one of the flux tubes about the footpoints of the other. 

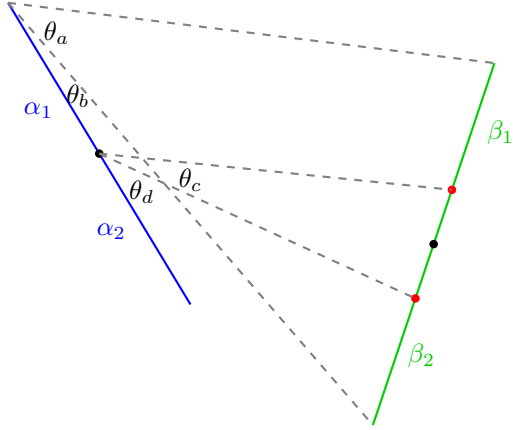
\begin{figure}[h]
\centering{
\begin{tikzpicture}[scale=0.8]
\draw[thick,blue] (-3,3)--(0,-2);
\draw[thick,black!20!green] (3,-4)--(5,2);

\fill (-1.5,0.5)  circle[radius=2pt];
\fill (4,-1)  circle[radius=2pt];

\fill[red] (3.7,-1.9)  circle[radius=2pt];
\fill[red] (4.3,-0.1)  circle[radius=2pt];

\draw[thick,gray, dashed] (-3,3)--(3,-4);
\draw[thick,gray,dashed] (-3,3)--(5,2);
\draw[thick,gray,dashed] (-1.5,0.5)--(3.7,-1.9);
\draw[thick,gray,dashed] (-1.5,0.5)--(4.3,-0.1);

\node at (-2.2,2.85) [anchor=north]{$\theta_a$};
\node at (-1.85,1.8) [anchor=north]{$\theta_b$};
\node at (0,0.4) [anchor=north]{$\theta_c$};
\node at (-0.8,0.2) [anchor=north]{$\theta_d$};

\node[blue] at (-2.5,1.5) [anchor=north]{$\alpha_1$};
\node[blue] at (-1.3,-0.5) [anchor=north]{$\alpha_2$};

\node[black!20!green] at (5.1,1.2) [anchor=north]{$\beta_1$};
\node[black!20!green] at (3.8,-2.5) [anchor=north]{$\beta_2$};

\end{tikzpicture}}
\caption{\label{pw_angles}The winding angles of $\beta$ relative to $\alpha_1$ for the pairwise winding calculation. All angles shown are between dashed lines emanating directly from $\alpha_1$ or between a dashed line and $\alpha$.}
\end{figure}

Proceeding to our alternative approach, Figure \ref{pw_angles} displays the winding angles for the full pairwise winding calculation $\cl(\alpha_1,\beta)$, using equation (\ref{wind2}). It is straightforward to identify, taking into account the direction of the changing angle and the $\sigma$-weightings of the curve sections, that $\cl(\alpha_1,\beta_1) = [(\theta_c+\theta_d)-(\theta_a+\theta_b)]/(2\pi)$, $\cl({\rm jump})=-\theta_c/(2\pi)$ and $\cl(\alpha_1,\beta_2) = (\theta_b-\theta_d)/(2\pi)$. Summing these terms, we find that
\be
\cl(\alpha_1,\beta) = -\frac{1}{2\pi}\theta_a=-\frac{1}{2\pi}\nu.
\en
We thus find the same result as equation (\ref{nu}). Performing a similar calculation for $\cl(\alpha_2,\beta)$ reproduces equation (\ref{rho}). Therefore, equation (\ref{wright_berger}) can be interpreted as depending on the winding of one tube about the footpoints of the other, as we showed first, or in terms of the more general description of the pairwise winding of the two tubes. 

\section{Self and mutual helicity flux}
The self and mutual helicities that we have considered until now are determined in a three-dimensional spatial domain (containing a magnetic field) at a particular instant in time. However,  the concept of mutual and self helicities can also be extended {to the temporal winding of field lines, an area that has received attention in other areas of fluid dynamics \cite{25}. In particular, the mutual/self decomposition can be extended}  to treat magnetic helicity flux through a planar boundary - a quantity calculated in a domain with two spatial dimensions and one temporal. The time-integrated flux of relative helicity through a planar boundary $P$ (see, for example, \cite{11,12,21}) is given by  
\be\label{hel_flux}
H_R^P = -\frac{1}{2\pi}\int_0^T\int_{P\times P}\frac{\d\theta(\xv,\yv)}{\d t}B_z(\xv)B_z(\yv)\,\d^2x\,\d^2y\,\d t,
\en
where the integration takes place in the time period $t\in[0,T]$. There is a clear similarity between equations (\ref{hel_flux}) and (\ref{wind_hel2}), with one converting to the other \emph{via} the transformation $[0,h]\leftrightarrow[T,0]$. The relative helicity through the plane $H_R^P$ can also be written as the winding helicity of a particular vector in $\Omega^t$, which is the same domain as $\Omega$ but with the vertical dimension replaced by time, i.e. a domain built from stacking planes $P$ at different times for $t\in[0,T]$ rather than stacking slices $S_z$ for $z\in[0,h]$. Vectors in this domain have the basis $\{\ev_x,\ev_y,\ev_t\}$, where the last unit vector represents the `time dimension'. {The topological interpretation of magnetic helicity flux is important to consider as often it is this quantity that can be determined in practice rather that the full three-dimensional helicity, such as in solar observations  \cite{21}.}  
\begin{theorem}
The winding helicity of the divergence-free vector $\Cv = \Ev\times\ev_t + B_z\ev_t$ in $\Omega^t$ is equal to minus the time-integrated helicity flux through a horizontal plane $P$, where $\Ev$ is the electric field of ideal MHD and $B_z=\Bv\cdot\ev_z$ (both measured on $P$).
\end{theorem}
\begin{proof}
We first show that $\Cv$ is divergence-free. In $\Omega^t$, the divergence operator can be written as
\be
{\rm div} = \frac{\partial}{\partial x} + \frac{\partial}{\partial y} + \frac{\partial}{\partial t}.
\en
Then, making use of Faraday's law,
\begin{eqnarray}
{\rm div}\,\Cv &=& \frac{\partial E_y}{\partial x} - \frac{\partial E_x}{\partial y} + \frac{\partial B_z}{\partial t} \nonumber \\
&=& \left(\nabla\times\Ev + \frac{\partial \Bv}{\partial t}\right)_z = 0.
\end{eqnarray}
We now show that $\Cv = \Ev\times\ev_t + B_z\ev_t$. After differentiating $\theta(\xv,\yv)$ in equation (\ref{theta}) with respect to $t$ and performing some simple manipulations, it can be shown that
\be\label{dtdt1}
\frac{\d}{\d t}\theta(\xv(t),\yv(t)) = \frac{1}{|\rv|^2}\left(r_1\frac{\d r_2}{\d t} - r_2\frac{\d r_1}{\d t}\right),
\en
where $\rv=\xv-\yv$ with $r_1=x_1-y_1$ and $r_2 = x_2-y_2$. As the magnetic field emerges through or moves laterally on $P$, the evolution of a vector $\xv$ to a point where a field line intersects $P$ is given by
\be\label{fieldline}
\frac{\d\xv}{\d t} = \uv^P-\frac{u_z}{B_z}\Bv^P,
\en
where $\uv$ is the velocity field and $\uv^P$ and $\Bv^P$ are the projections onto $P$ of the velocity and magnetic fields respectively \cite{21}. Inserting equation (\ref{fieldline}) into equation (\ref{dtdt1}) leads to
\be\label{c1}
\frac{\d}{\d t}\theta(\xv(t),\yv(t))= \frac{1}{|\rv|^2}\left[r_1\left(\frac{E_y(\xv)}{B_z(\xv)}-\frac{E_y(\yv)}{B_z(\yv)}\right)-r_2\left(\frac{E_x(\xv)}{B_z(\xv)}-\frac{E_x(\yv)}{B_z(\yv)}\right)\right],
\en
where $\Ev=-\uv\times\Bv$.

Writing the vector $\Cv$ in terms of projections parallel to $P$, $\Cv^P$, and perpendicular to $P$, $C_t\ev_t$, we follow the procedure in \cite{10} to find
\be\label{c2}
\Av^W[\Bv(\xv)]\cdot\Bv(\yv) = \frac{1}{2\pi}\int_P\left(\frac{\Cv^P(\xv)}{C_t(\xv)}-\frac{\Cv^P(\yv)}{C_t(\yv)}\right)\cdot\frac{\ev_t\times\rv}{|\rv|^2}C_t(\xv)C_t(\yv)\,\d^2y.
\en 
Since $\ev_t\times\rv = (-r_2,r_1,0)^T$, a comparison of equations (\ref{c1}) and (\ref{c2}) reveals that the components of $\Cv$ are $C_x = E_y$, $C_y=E_x$ and $C_t=B_z$. In other words,
\be
\Cv = \Ev\times\ev_t+B_z\ev_t
\en
and
\be
\int_{\Omega^t}\Av^W(\Cv)\cdot\Cv\,\d V = -H^P_R.
\en
\end{proof}
  
\begin{corollary}\label{cor}
The time-integrated helicity flux has a self-mutual decomposition for $\Omega^t=\sqcup_i\Omega^t_i$,
\be
    H_R^P=\sum_i H_{Ri}^P + \sum_{{i,j \atop (i\ne j)}}H^P_{Rij},
\en
where the self helicities are
\be\label{self_time}
H^P_{Ri}=-\frac{1}{2\pi}\int_0^T\int_{{P\cap\Omega^t_i\times \atop P\cap\Omega^t_i}}\frac{\d}{\d t}\theta(\xv,\yv)B_z(\xv)B_z(\yv)\,\d^2x\,\d^2y\,\d t,
\en
and the mutual helicities are
\be\label{mutual_time}
H^P_{Rij}=-\frac{1}{2\pi}\int_0^T\int_{{P\cap\Omega^t_i\times \atop P\cap\Omega^t_j}}\frac{\d}{\d t}\theta(\xv,\yv)B_z(\xv)B_z(\yv)\,\d^2x\,\d^2y\,\d t,
\en
and satisfy $H^P_{Rij}=H^P_{Rji}$.
\end{corollary}  
The proof of Corollary \ref{cor} is very similar to that of Theorem \ref{sm_thrm} and is omitted for brevity. Similar to Theorem \ref{sm_thrm}, the position vectors $\xv$ and $\yv$ point to locations in the same subdomain for self helicity, but in different subdomains for mutual helicity. As indicated by the existence of a self-mutual decomposition, there is a deep connection between winding and helicity flux. This connection is explored in detail in \cite{11,12,13}.  

 

\section{Simulation example}

As a demonstration of the practical calculation of self and mutual winding helicities, we apply equations (\ref{self}) and (\ref{mutual}) to the magnetic field described in \cite{22}. The magnetic field consists, originally, of three domed subdomains where the magnetic field only connects to the lower plane (closed field lines). Outside of these domed regions, there is a subdomain where the magnetic field connects from the lower to the upper boundary (open field lines). This magnetic field structure can be seen in Figure \ref{fig: streamlines}(a). The magnetic field is deformed by applying a horizontal flow on the lower plane, which is a time-periodic driver. The driver causes the deformation of the domes  and the transfer of twist upwards into the domain, see Figure \ref{fig: streamlines}(b). Reconnection (dome to dome and dome to open field) leads to the number of domes changing over time, until they are completely destroyed. Full details of the setup and evolution can be found in \cite{22}. Here we will focus on describing the magnetic field evolution with reference to self and mutual winding helicity.

\begin{figure*}[t!]\begin{center}
\begin{tikzpicture}
\node (t0) at (0,6.5) {\includegraphics[width=0.8\columnwidth]{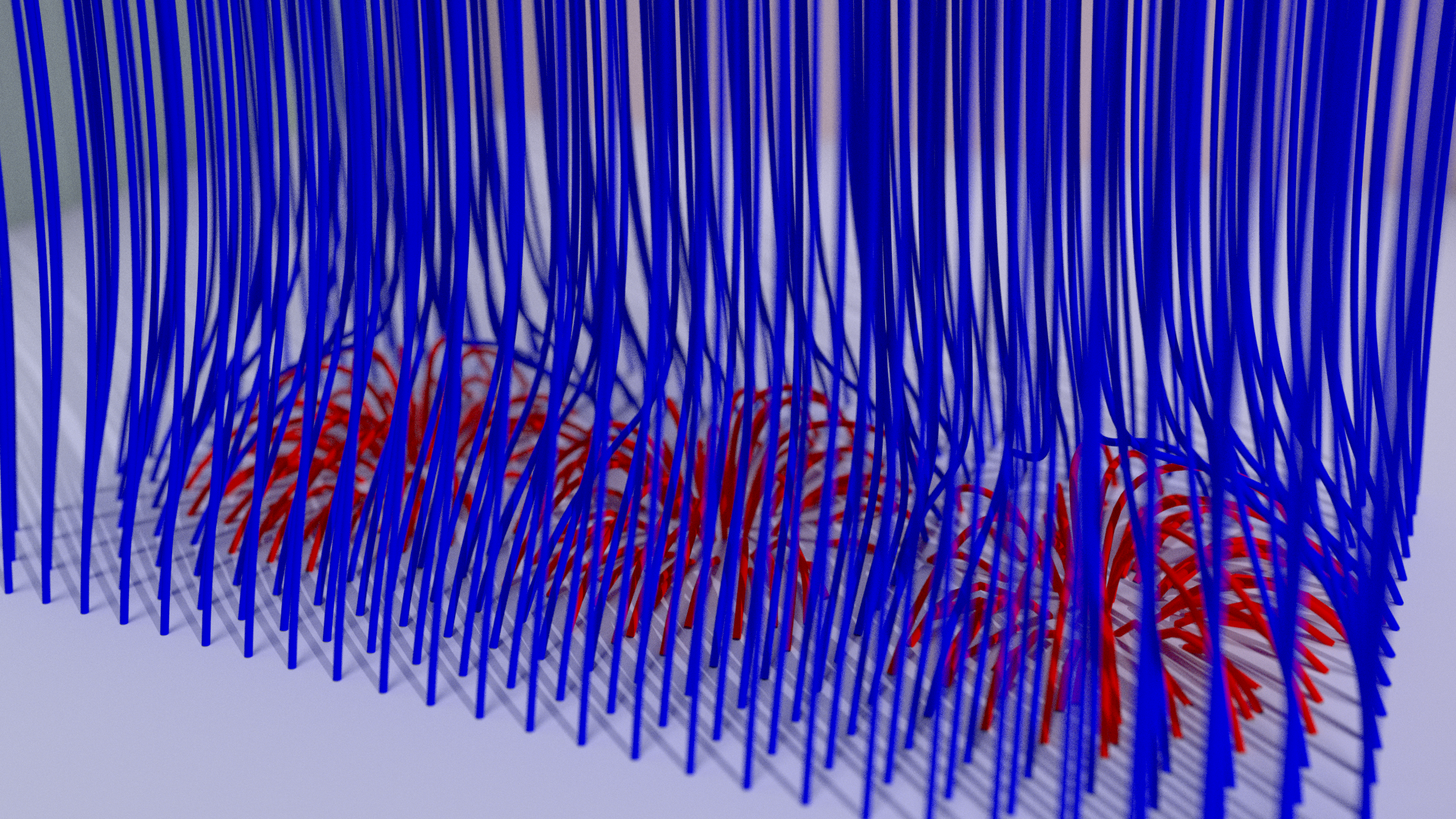}};
\node (t50) at (0,0) {\includegraphics[width=0.8\columnwidth]{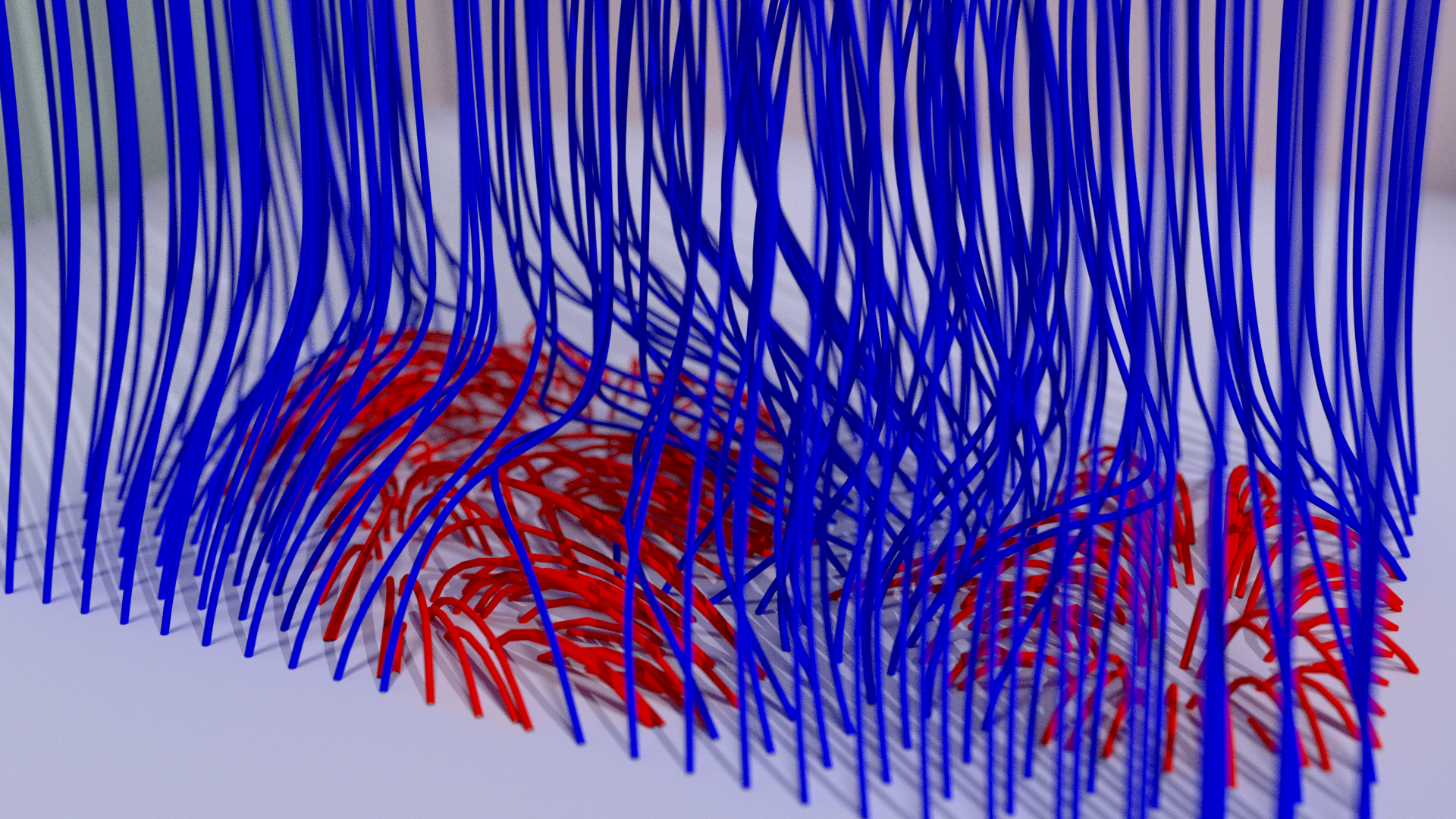}};
\node at (-5,4){\large (a)};
\node at (-5,-2.5){\large (b)};
\end{tikzpicture}
\end{center}
\caption[]{\label{fig: streamlines}
Magnetic field line plots for the simulation example.
Blue field lines are open, while red are closed.
We show the initial configuration in panel (a) and the configuration
at time $t = 50$ in panel (b).
}
\end{figure*}

 At each time step, we determine the locations of the open and closed subdomains by simply plotting numerous field lines, which gives us a good estimate of the open-closed field boundaries.  We calculate, by means of equations (\ref{self}) and (\ref{mutual}), the self winding helicity of the open subdomain $H^{\rm O}_{\rm self}$, the self winding helicity of the combined closed subdomains $H^{\rm C}_{\rm self}$, the total mutual winding helicity of the open subdomain with all the closed subdomains $H^{\rm OC}_{\rm mutual}$ and the total mutual winding helicity of the closed subdomains $H^{\rm C}_{\rm mutual}$. Figure \ref{fig: helicities} displays the time series for all the helicities listed above. The driver on the lower boundary of the domain twists the magnetic field over a period of $t=32$. Positive (anticlockwise) and negative (clockwise) twists are applied in an alternating fashion. 

There is an initial positive twist which acts to deform two of the domed regions. This leads to an increase in the self helicity $H^{\rm C}_{\rm self}$ of the closed magnetic field. Due to the nature of the driver applied \cite{22}, the mutual helicity $H_{\rm mutual}^{\rm C}$ remains very small. At $t\approx23$, reconnection causes a transference of the self helicity $H^{\rm C}_{\rm self}$ in the domes to a sharp increase in the mutual heilcity between the domes $H^{\rm C}_{\rm mutual}$. The open field lines are also deformed, both directly by the driver and by the deformation of the domes. There is a large input of $H^{\rm O}_{\rm self}$, whose negative values indicate that, overall, the open field lines are winding in the opposite sense to those in the domes. As both the open and domed subdomains deform, the mutual helicity $H^{\rm OC}_{\rm mutual}$ between them also increases. The discrepancy in the size of the helicities in the open and closed subdomains can be put down to how the field lines are connected in these domains. The closed field lines have high curvature, and so have high magnetic tension. It is more difficult to deform these regions and increase the winding of field lines compared to the open subdomain, where the tension of field lines is much weaker.

Subsequent twists by the driver lead to cyclic behaviour in all of the helicities. Only $H^{\rm O}_{\rm self}$ increases in magnitude with time, whereas the rest decrease to zero by $t\approx150$, when the domes cease to exist. Hence the decay of all the helicities involving the closed subdomains. The twist injected into the domain by the driver transfers more and more to the open field. After the destruction of the domes, the open field occupies the entire domain.

\begin{figure*}[t!]\begin{center}
\includegraphics[width=0.75\columnwidth]{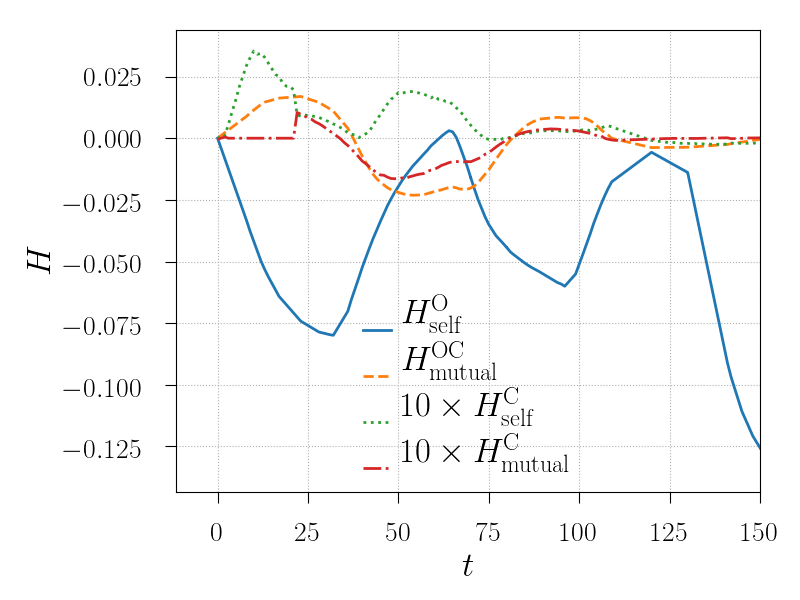}
\end{center}
\caption[]{\label{fig: helicities}
Time evolution of the self and mutual helicities from the
test simulation consisting of initially three separatrix domes
that are being periodically driven in time. The deformation of domes leads to their destruction at $t\approx 150$, when all of the helicity is due to the open magnetic field.
}
\end{figure*}

\section{Conclusions}
We have shown that the winding helicity representation of  relative helicity, for magnetic fields between two horizontal planes, allows for a self-mutual decomposition, similar to that of classical helicity. Our results generalize previous expressions for the self and mutual relative helicity of thin and discrete flux tubes. We further show that the self-mutual decomposition can be extended to a space-time domain, thus showing that helicity flux can also have such a decomposition. We illustrate the practical calculation of self and mutual helicity in a resistive MHD simulation.

Our self-mutual helicity decomposition is based on winding helicity. Thus, the underlying (field line) topological structure of both self and mutual helicity is the average pairwise winding of magnetic field lines. In other studies, such as \cite{23}, self (relative) helicities can be found in subdomains, but they require reference magnetic fields to be calculated in these subdomains. Depending on the topology of the subdomain itself, this approach can be challenging to apply in practice and  make understanding the field line topology, based on the relative helicity, difficult. Winding helicity avoids these problems and gives a clear interpretation for the field line topology in every subdomain. {This property is expected to be important in the analysis of models of solar coronal mass ejections, where self and mutual winding helicities can be used to understand, and perhaps predict, the development and onset of eruptions (see \cite{26} for a typical set-up). The application of self and mutual winding helicities to solar eruptions will be developed in future research.}  Due to the expressions of self and winding helicity being relatively simple to calculate, we expect these to become a useful tool in the analysis of magnetic topology in MHD simulations.

\end{document}